\newif\iflipics
\lipicsfalse

\iflipics
\documentclass[a4paper,USenglish]{lipics-v2019}
\else
\documentclass[11pt]{article}
\fi

\usepackage{graphicx}
\usepackage{xspace,amsmath,amssymb,amsthm,mathtools}
\usepackage{algorithmicx,algorithm}
\usepackage{cite}
\usepackage[noend]{algpseudocode}
\usepackage{multirow}
\usepackage{listings}
\iflipics
\else
\usepackage[colorlinks=true]{hyperref}
\usepackage{fullpage}
\fi

\algnewcommand\algorithmicinput{\textbf{INPUT:}}
\algnewcommand\INPUT{\item[\algorithmicinput]}
\algnewcommand\algorithmicoutput{\textbf{OUTPUT:}}
\algnewcommand\OUTPUT{\item[\algorithmicoutput]}

\newcommand{\problem}[1]{\textsc{#1}\xspace}

\iflipics
\else

\newtheorem{lemma}{Lemma}

\fi

\iflipics
\title{Improved Analysis of Highest-Degree Branching for Feedback Vertex Set}
\titlerunning{Improved Analysis of Highest-Degree Branching for Feedback Vertex Set}
\fi

\iflipics
\author{Yoichi Iwata}{National Institute of Informatics, Japan}{yiwata@nii.ac.jp}{}{Supported by JSPS KAKENHI Grant Number JP17K12643}
\author{Yusuke Kobayashi}{Kyoto University, Japan}{yusuke@kurims.kyoto-u.ac.jp}{}{Supported by JSPS KAKENHI Grant Numbers JP16K16010, 17K19960, and 18H05291}
\authorrunning{Y. Iwata and Y. Kobayashi}
\Copyright{Yoichi Iwata and Yusuke Kobayashi}
\ccsdesc{Design and analysis of algorithms~Fixed parameter tractability}
\keywords{Feedback Vertex Set, Branch and bound, Measure and conquer}
\EventEditors{John Q. Open and Joan R. Access}
\EventNoEds{2}
\EventLongTitle{42nd Conference on Very Important Topics (CVIT 2016)}
\EventShortTitle{CVIT 2016}
\EventAcronym{CVIT}
\EventYear{2016}
\EventDate{December 24--27, 2016}
\EventLocation{Little Whinging, United Kingdom}
\EventLogo{}
\SeriesVolume{42}
\ArticleNo{23}
\else
\title{Improved Analysis of Highest-Degree Branching\\ for Feedback Vertex Set}
\author{Yoichi Iwata\thanks{Supported by JSPS KAKENHI Grant Number JP17K12643}\\
	National Institute of Informatics, Japan\\
	\texttt{yiwata@nii.ac.jp}
\and
    Yusuke Kobayashi\thanks{Supported by JSPS KAKENHI Grant Numbers JP16K16010, 17K19960, and 18H05291}\\
    Kyoto University, Japan\\
    \texttt{yusuke@kurims.kyoto-u.ac.jp}
}
\fi

\date{}

\begin{document}
\maketitle

\begin{abstract}
Recent empirical evaluations of exact algorithms for \problem{Feedback Vertex Set} have demonstrated the efficiency of a highest-degree branching algorithm with a degree-based pruning heuristic.
In this paper, we prove that this empirically fast algorithm runs in $O(3.460^k n)$ time, where $k$ is the solution size.
This improves the previous best $O(3.619^k n)$-time deterministic algorithm obtained by Kociumaka and Pilipczuk.
\end{abstract}

\iflipics
\newpage
\fi

\section{Introduction}
\problem{Feedback Vertex Set (FVS)} is a classical NP-hard graph optimization problem of finding the minimum-size vertex deletion set to make the input graph a forest.
It is known that this problem is \emph{fixed-parameter tractable (FPT)} parameterized by the solution size $k$; i.e., we can find a deletion set of size $k$ in $O^*(f(k))$\footnote{$O^*(\cdot)$ hides factors polynomial in $n$. Note that for \problem{FVS}, any $O(f(k) n^{O(1)})$-time FPT algorithms can be improved to $O(f(k)k^{O(1)}+k^{O(1)}n)$ time by applying a linear-time kernel~\cite{conf/icalp/Iwata17} as a preprocess.
We can therefore focus only on the $f(k)$ factor when comparing the running time.} time for some function $f$.
FVS is one of the most comprehensively studied problems in the field of parameterized algorithms, and various FPT algorithms using different approaches have been developed, including short-cycle branching~\cite{conf/dagstuhl/DowneyF92}, highest-degree branching~\cite{conf/soda/Cao18}, iterative-compression branching~\cite{journals/jcss/ChenFLLV08,journals/algorithmica/CaoC015,journals/ipl/KociumakaP14}, LP-guided branching~\cite{journals/corr/Wahlstrom13,conf/focs/IwataYY18}, cut-and-count dynamic programming~\cite{conf/focs/CyganNPPRW11}, and random sampling~\cite{journals/jair/BeckerBG00}.

The current fastest deterministic FPT algorithm for \problem{FVS} is a branching algorithm combined with the iterative compression technique~\cite{journals/ipl/KociumakaP14} which runs in $O^*(3.619^k)$ time.
When allowing randomization, the current fastest one is a cut-and-count dynamic programming algorithm~\cite{conf/focs/CyganNPPRW11} which runs in $O^*(3^k)$ time.
In this paper, we give a faster deterministic algorithm which runs in $O^*(3.460^k)$ time.
As explained below, this study is strongly motivated by Parameterized Algorithms and Computational Experiments (PACE) challenge and its follow-up empirical evaluation by Kiljan and Pilipczuk~\cite{conf/wea/KiljanP18}.
Instead of designing a new theoretically fast algorithm, we analyze the theoretical worst-case running time of the empirically fast algorithm that has been developed through the PACE challenge and the empirical evaluation, and we show that this algorithm is not only empirically fast but also theoretically fast.

PACE challenge is an annual programming challenge started in 2016.
Due to its importance in the field, \problem{FVS} was selected as the subject of track B in the first PACE challenge~\cite{conf/iwpec/DellHJKKR16}.
Although in theoretical studies, the current fastest algorithm is the randomized cut-and-count dynamic programming, the result of the challenge suggests that branching is the best choice in practice.
This is not so surprising; because the theoretical analysis of branching algorithms is difficult, the proved upper bound of the running time is rather pessimistic.

Among seven submissions to the PACE challenge, top six submissions used branching algorithms;
the first used the LP-guided branching;
the second and third used branching on highest-degree vertices;
the fourth and sixth used branching combined with iterative compression;
and the fifth used branching on short cycles.
In addition to the pruning by the LP lower bound, the first-place solver by Imanishi and Iwata~\cite{github/Iwata16} used the following \emph{degree-based pruning} heuristic:
\begin{lemma}[\cite{github/Iwata16}]\label{lem:pruning-orig}
Given a set of undeletable vertices $F\subseteq V$, let $v_1,v_2,\ldots$ be the vertices of $V(G)\setminus F$ in the non-decreasing order of the degrees $d(v_i)$ in $G$.
If $|E(G)|-\sum_{i=1}^k d(v_i)\geq |V(G)|-k$ holds, there is no feedback vertex set $S\subseteq V(G)\setminus F$ of size $k$.
\end{lemma}

The follow-up empirical evaluation~\cite{conf/wea/KiljanP18} shows that the use of the degree-based pruning is much more important than the choice of branching rules.
By combining with the degree-based pruning, the performances of the LP-guided branching~\cite{journals/corr/Wahlstrom13}, the highest-degree branching~\cite{conf/soda/Cao18}, and the iterative-compression branching~\cite{journals/ipl/KociumakaP14}, are all significantly improved, and among them, the highest-degree branching slightly outperforms the others.
Cao~\cite{conf/soda/Cao18} showed that one can stop the highest-degree branching at depth $3k$ by using a degree-based argument, and therefore the running time is $O(8^k)$.
On the other hand, the theoretically proved running time of other branching algorithms (without the degree-based pruning) are, $O^*(4^k)$ for the LP-guided branching~\cite{journals/corr/Wahlstrom13} and $O^*(3.619^k)$ for the iterative-compression branching~\cite{journals/ipl/KociumakaP14}.
These affairs motivated us to refine the analysis of the highest-degree branching with the degree-based pruning.

In our analysis, instead of bounding the depth of the search tree as Cao~\cite{conf/soda/Cao18} did, we design a new measure to bound the size of the search tree.
The measure is initially at most $k$ and we show that the measure drops by some amount for each branching.
In contrast to the standard analysis of branching algorithms, our measure has a negative term and thus can have negative values; however, we show that we can immediately apply the degree-based pruning for all such cases.
A simple analysis already leads to an $O^*(4^k)$-time upper bound which significantly improves the $O^*(8^k)$-time upper bound obtained by Cao~\cite{conf/soda/Cao18}.
We then apply the measure-and-conquer analysis~\cite{journals/jacm/FominGK09} and improve the upper bound to $O^*(3.460^k)$.

\subsection{Organization}
Section~\ref{sec:alg} describes the highest-degree branching algorithm with the degree-based pruning.
In Section~\ref{sec:analysis}, we analyze the running time of the algorithm.
We first give a simple analysis in Section~\ref{sec:simple} and then give a measure-and-conquer analysis in Section~\ref{sec:mc}.
While the correctness of the simple analysis can be easily checked, we need to evaluate thousands of inequalities to check the correctness of the measure-and-conquer analysis.
For convenience, we attach a source code of the program to evaluate the inequalities in Appendix~\ref{sec:program}.
The same source code is also available at \url{https://github.com/wata-orz/FVS_analysis}.

\section{Algorithm}\label{sec:alg}

An input to the algorithm is a tuple $(G,F,k)$ of a multi-graph $G=(V,E)$, a set of undeletable vertices $F\subseteq V$, and an integer $k$.
Our task is to find a subset of vertices $S\subseteq V\setminus F$ such that $|S|\leq k$ and $G[V\setminus S]$ contains no cycles.
Note that a double edge is also considered as a cycle.
We denote by $N(u)$ the multiset of the adjacent vertices of $u$ and define $d(u):=|N(u)|$.
For convenience, we use $D$ to denote $\max_{v\in V\setminus F}d(v)$.

Our algorithm uses the standard reduction rules listed below in the given order (i.e., rule $i$ is applied only when none of the rules $j$ with $j<i$ are applicable).
All of these reductions are also used in the empirical evaluation by Kiljan and Pilipczuk~\cite{conf/wea/KiljanP18}.
If none of the reductions are applicable, we apply a pruning rule, and if it cannot be pruned, we apply a branching rule.

\paragraph*{Reduction Rule 1.}
If there exists a vertex $u$ of degree at most one, delete $u$.

\paragraph*{Reduction Rule 2.}
If there exists a vertex $u\not\in F$ such that $G[F\cup\{u\}]$ contains a cycle, delete $u$ and decrease $k$ by one.

\paragraph*{Reduction Rule 3.}
If there exists a vertex $u$ of degree two, delete $u$ and add an edge connecting its two endpoints.

\paragraph*{Reduction Rule 4.}
If there exists an edge $e$ of multiplicity more than two, reduce its multiplicity to two.

\paragraph*{Reduction Rule 5.}
If there exists a vertex $u\not\in F$ incident to a double edge $uw$ with $d(w)\leq 3$, delete $u$ and decrease $k$ by one.

\paragraph*{Reduction Rule 6.}
If $D\leq 3$, solve the problem in polynomial time by a reduction to the matroid matching~\cite{journals/algorithmica/CaoC015}.

\paragraph*{Pruning Rule.}
If $k<0$ or $kD-\sum_{v\in F}(d(v)-2)<0$, return NO.

\paragraph*{Branching Rule.}
Pick a vertex $u\in V\setminus F$ of the highest degree $D$.
Let $U:=N(u)\cap F$ and let $G'$ be the graph obtained by contracting $U\cup\{u\}$ into a single vertex $u'$.
We branch into two cases: $(G-u,F,k-1)$ and $(G',F-U+u',k)$.

\paragraph*{}
The correctness of the first four reduction rules is trivial. We can prove the correctness of the reduction rule 5 as follows.
Because there is a double edge $uw$, any feedback vertex set must contain at least one of $u$ and $w$.
Because $w$ has at most one edge other than the double edge $uw$, every cycle containing $w$ also contains $u$.
Therefore, there always exists a minimum feedback vertex set containing $u$.

After applying the reductions, the following conditions hold.
\begin{enumerate}
    \item $G$ has the minimum degree at least three.
    \item No double edges are incident to $F$.
    \item For any vertex $v\not\in F$, $G-v$ has the minimum degree at least two.
    \item $D\geq 4$.
\end{enumerate}

The correctness of the pruning rule follows from the following lemma.
\begin{lemma}\label{lem:pruning}
	If the minimum degree of $G$ is at least two and $kD-\sum_{v\in F}(d(v)-2)<0$ holds, there is no feedback vertex set $S\subseteq V\setminus F$ of size at most $k$.
\end{lemma}

\begin{proof}
	Suppose that there is a feedback vertex set $S\subseteq V\setminus F$ of size at most $k$. We have
	\begin{align*}
		kD-\sum_{v\in F}(d(v)-2)&\geq \sum_{v\in S}d(v)-\sum_{v\in V\setminus S}(d(v)-2)\\
		&= \sum_{v\in S}d(v)-\sum_{v\in V\setminus S}d(v)+2|V\setminus S|\\
		&= 2|E(S)|-2|E(V\setminus S)|+2|V\setminus S|\\
		&\geq -2|E(V\setminus S)|+2|V\setminus S|.
	\end{align*}
	Because $G-S$ is a forest, this must be non-negative.
\end{proof}

Note that this pruning is different from the degree-based pruning (Lemma~\ref{lem:pruning-orig}) used in the Imanishi-Iwata solver~\cite{github/Iwata16} and the empirical evaluation~\cite{conf/wea/KiljanP18};
however, as the following lemma shows, if this pruning is applied, then the original degree-based pruning is also applied.
Therefore, we can use the same analysis against the original degree-based pruning.
We use this weaker version because it is sufficient for our analysis.
We leave whether the stronger version helps further improve the analysis as future work.

\begin{lemma}
For a subset $F\subseteq V$, let $v_1,v_2,\ldots$ be the vertices of $V\setminus F$ in the non-increasing order of the degrees $d(v_i)$ in $G$.
If the minimum degree of $G$ is at least two, $kd(v_1)-\sum_{v\in F}(d(v)-2)<0$ implies $|E|-\sum_{i=1}^k d(v_i)\geq |V|-k$.
\end{lemma}
\begin{proof}
Let $S=\{v_1,\ldots,v_k\}$ and assume that $kd(v_1)-\sum_{v\in F}(d(v)-2)<0$.
Then, we have
\begin{align*}
2\left(|V|-k-|E|+\sum_{v\in S}d(v)\right)&=\sum_{v\in S}d(v)-\left(2|E|-\sum_{v\in S}d(v)-2(|V|-k)\right)\\
&=\sum_{v\in S}d(v)-\sum_{v\in V\setminus S}(d(v)-2)\\
&\leq kd(v_1) - \sum_{v\in F}(d(v)-2) < 0.\qedhere
\end{align*}
\end{proof}

\section{Analysis}\label{sec:analysis}

For parameters $0\leq\alpha\leq 1$ and $\beta_d$ satisfying $0=\beta_0=\beta_1=\beta_2\leq\beta_3\leq\beta_4\leq \cdots$, we define
\[
	\mu(G,F,k):=k-\frac{\alpha}{D}\sum_{v\in F}(d(v)-2)+\sum_{v\in F}\beta_{d(v)}.
\]
Initially, we have $\mu(G,\emptyset, k)=k$.

\begin{lemma}
	After the pruning, we have $\mu(G,F,k)\geq 0$.
\end{lemma}
\begin{proof}
\[
		\mu(G,F,k)=(1-\alpha)k+\frac{\alpha}{D}\left(kD-\sum_{v\in F}(d(v)-2)\right)+\sum_{v\in F}\beta_{d(v)}\geq 0.\qedhere
\]
\end{proof}

We now show that applying the reduction rules does not increase $\mu$.
We can easily see that $D$ never increases by the reduction but may decrease; however, because such decrease leads to a smaller $\mu$, we can analyze as if $D$ does not change by the reduction.
Because the reduction rule 3 deletes a vertex of degree two and does not change the degrees of other vertices, it does not change $\mu$.
Because the reduction rule 4 is applied only when the reduction rule 2 cannot be applied, it does not change the degrees of vertices in $F$, and therefore it does not change $\mu$.
Because the graph immediately after branching has the minimum degree at least two, we apply the reduction rule 1 only after applying the reduction rules 2 or 5.

\begin{lemma}
The reduction rule 2 or 5, together with the subsequent applications of the reduction rule 1, does not increase $\mu$.
\end{lemma}
\begin{proof}
By deleting an edge $vw$ such that $v\not\in F$ and $w\in F$, $\mu$ increases by at most
\[
    \frac{\alpha}{D}-(\beta_{d(w)}-\beta_{d(w)-1})\leq \frac{\alpha}{D}\leq \frac{1}{D}.
\]
In the series of the reductions, the number of such deletion is at most $d(u)\leq D$.
Because $k$ decreases by one, the increase of $\mu$ is at most $-1+D\times \frac{1}{D}\leq 0$.
\end{proof}

Finally, we analyze the branching rule.
Let $f:=|U|$ and $\mathbf{d}:=\{d_1,\ldots,d_f\}$ be the multiset of degrees of vertices in $U$.
The degree of $u'$ is $d':=D+\sum_{i=1}^f (d_i-2)$.
In the former case of the branching, we have
\begin{align*}
    \Delta_1(D,\mathbf{d})&:=\mu(G,F,k)-\mu(G-u,F,k-1)\\
    &=1-f\frac{\alpha}{D}+\sum_{i=1}^{f}(\beta_{d_i}-\beta_{d_i-1}).
\end{align*}
In the latter case, we have
\begin{align*}
    \Delta_2(D,\mathbf{d})&:=\mu(G,F,k)-\mu(G',F-U+u',k)\\
    &=-\frac{\alpha}{D}\left(\sum_{i=1}^f (d_i-2)-(d'-2)\right)+\sum_{i=1}^f \beta_{d_i}-\beta_{d'}\\
    &=\alpha-2\frac{\alpha}{D}+\sum_{i=1}^f \beta_{d_i}-\beta_{d'}.
\end{align*}

If $c^{-\Delta_1(D,\mathbf{d})}+c^{-\Delta_2(D,\mathbf{d})}\leq 1$ holds for any $(D,\mathbf{d})$ for some $c>1$, the running time of the algorithm is bounded by $O^*(c^k)$.
We now optimize the parameters to minimize $c$.

\subsection{Simple Analysis}\label{sec:simple}

As a simple analysis whose correctness can be easily checked, we use $\alpha=\log_4\frac{8}{3}\approx 0.7075$, $\beta_d=\frac{1}{2}\log_4 \frac{3}{2}\approx 0.1462$ for all $d\geq 3$, and $c=4$.
Note that $D\geq 4$ holds.
For these parameters, we have
\begin{align*}
c^{-\Delta_1(D,\mathbf{d})}&\leq 4^{-1+\frac{f}{D}\log_4\frac{8}{3}}=\frac{1}{4}\cdot\left(\frac{8}{3}\right)^{\frac{f}{D}}\leq\frac{1}{4}\min\left(\frac{8}{3},\left(\frac{8}{3}\right)^{\frac{f}{4}}\right),\\
c^{-\Delta_2(D,\mathbf{d})}&=4^{\left(\frac{2}{D}-1\right)\log_4\frac{8}{3}+(1-f)\frac{1}{2}\log_4\frac{3}{2}}\leq \left(\frac{3}{8}\right)^{\frac{1}{2}}\cdot \left(\frac{3}{2}\right)^{\frac{1-f}{2}}.
\end{align*}

We now show that $c^{-\Delta_1(D,\mathbf{d})}+c^{-\Delta_2(D,\mathbf{d})}\leq 1$ holds by the following case analysis.

\begin{align*}
c^{-\Delta_1(D,\mathbf{d})}+c^{-\Delta_2(D,\mathbf{d})}\leq\begin{dcases}
\frac{1}{4}+\left(\frac{3}{8}\right)^{\frac{1}{2}}\cdot\left(\frac{3}{2}\right)^\frac{1}{2}=1&(f=0),\\
\frac{1}{4}\cdot \left(\frac{8}{3}\right)^\frac{1}{4}+\left(\frac{3}{8}\right)^{\frac{1}{2}}<0.932&(f=1),\\
\frac{1}{4}\cdot \left(\frac{8}{3}\right)^\frac{2}{4}+\left(\frac{3}{8}\right)^{\frac{1}{2}}\cdot\left(\frac{2}{3}\right)^{\frac{1}{2}}<0.909&(f=2),\\
\frac{1}{4}\cdot \left(\frac{8}{3}\right)^\frac{3}{4}+\left(\frac{3}{8}\right)^{\frac{1}{2}}\cdot\frac{2}{3}<0.930&(f=3),\\
\frac{1}{4}\cdot\frac{8}{3}+\left(\frac{3}{8}\right)^{\frac{1}{2}}\cdot\left(\frac{2}{3}\right)^{\frac{3}{2}}=1&(f\geq 4).
\end{dcases}
\end{align*}

\subsection{Measure-and-Conquer Analysis}\label{sec:mc}
We use the parameters $\alpha=0.922863$, $\beta$ shown in Table~\ref{tab:beta}, and $c=3.460$.

\begin{table}[ht]
    \centering
    \begin{tabular}{c|c}
$d$ & $\beta_d$ \\
\hline
1 & 0.000000\\
2 & 0.000000\\
3 & 0.114038\\
4 & 0.186479\\
5 & 0.238143\\
6 & 0.277239\\
7 & 0.308030\\
8 & 0.332974\\
9 & 0.353536\\
10 & 0.370540\\
\end{tabular}
\begin{tabular}{c|c}
$d$ & $\beta_d$ \\
\hline
11 & 0.384771\\
12 & 0.396884\\
13 & 0.408715\\
14 & 0.418855\\
15 & 0.427643\\
16 & 0.435333\\
17 & 0.442118\\
18 & 0.448149\\
19 & 0.453544\\
20 & 0.458401\\
\end{tabular}
\begin{tabular}{c|c}
$d$ & $\beta_d$ \\
\hline
21 & 0.462794\\
22 & 0.466788\\
23 & 0.470435\\
24 & 0.473778\\
25 & 0.476853\\
26 & 0.479691\\
27 & 0.482320\\
28 & 0.484760\\
29 & 0.487032\\
$\geq 30$ & 0.489153\\
    \end{tabular}
    \caption{The values of $\beta$.}
    \label{tab:beta}
\end{table}

\begin{lemma}\label{lem:mc}
$c^{-\Delta_1(D,\mathbf{d})}+c^{-\Delta_2(D,\mathbf{d})}\leq 1$ holds for any $(D,\mathbf{d})$ with $D \ge 4$.
\end{lemma}
\begin{proof}
Suppose that $d_j\geq 32$ holds for some $j$.
Because $\beta_d=\beta_{30}$ for all $d\geq 30$ and because $d'=D+\sum_i (d_i-2)\geq D+(d_j-2)\geq 31$ holds, decreasing $d_j$ by one does not change $\Delta_1(D,\mathbf{d})$ nor $\Delta_2(D,\mathbf{d})$.
Therefore, we can focus on the case of $d_i\leq 31$ for all $i$.
We now show that the inequality holds by induction on $D$.
When $D=4$, we can verify that
\begin{align}
    c^{-\Delta_1(4,\mathbf{d})}+c^{-\Delta_2(4,\mathbf{d})}\leq 1 \enspace (\forall \mathbf{d})\label{eq:mc1}
\end{align}
holds by naively enumerating all the possible configurations of $\mathbf{d}$.
Assume that, for a fixed $D$, the inequality holds for any $(D-1,\mathbf{d})$.
We show that the inequality also holds for any $(D,\mathbf{d})$.

When $f<D$, we have
\[
    \Delta_1(D,\mathbf{d})\geq\Delta_1(D-1,\mathbf{d})
\]
and
\[
    \Delta_2(D,\mathbf{d})=\Delta_2(D-1,\mathbf{d})-2\frac{\alpha}{D}-\beta_{d'}+2\frac{\alpha}{D-1}+\beta_{d'-1},
\]
where $d'=D+\sum_{i=1}^f(d_i-2)$.
We can verify that our parameters satisfy
\begin{align}
    -2\frac{\alpha}{D}-\beta_{d'}+2\frac{\alpha}{D-1}+\beta_{d'-1}\geq 0 \enspace (\forall (D, d') \text{ with } 5\leq D\leq d'\leq 31).\label{eq:mc2}
\end{align}
Therefore, we have $\Delta_2(D,\mathbf{d})\geq\Delta_2(D-1,\mathbf{d})$.
This shows that 
\begin{align*}
    c^{-\Delta_1(D,\mathbf{d})}+c^{-\Delta_2(D,\mathbf{d})}\leq c^{-\Delta_1(D-1,\mathbf{d})}+c^{-\Delta_2(D-1,\mathbf{d})}\leq 1. 
\end{align*}

When $f=D$, let $\mathbf{d}':=\{d_1,\ldots,d_{f-1}\}$.
We have
\[
    \Delta_1(D,\mathbf{d})=\Delta_1(D-1,\mathbf{d}')+\beta_{d_f}-\beta_{d_{f-1}}\geq\Delta_1(D-1,\mathbf{d}')
\]
and
\begin{align*}
    \Delta_2(D,\mathbf{d})&=\Delta_2(D-1,\mathbf{d}')-2\frac{\alpha}{D}+2\frac{\alpha}{D-1}+\beta_{d_f}-\beta_{d'}+\beta_{d'-d_f+2}\\
    &\geq\Delta_2(D-1,\mathbf{d}')+\beta_{d_f}-\beta_{d'}+\beta_{d'-d_f+2}.
\end{align*}
Because we can verify that our parameters satisfy
\begin{align}
    \beta_{d_f}-\beta_{d'}+\beta_{d'-d_f+2}\geq 0 \enspace (\forall (d',d_f) \text{ with } 3\leq d_f<d'\leq 31),\label{eq:mc3}
\end{align}
we have $\Delta_2(D,\mathbf{d})\geq\Delta_2(D-1,\mathbf{d}')$.
This shows that 
\begin{align*}
    c^{-\Delta_1(D,\mathbf{d})}+c^{-\Delta_2(D,\mathbf{d})}\leq c^{-\Delta_1(D-1,\mathbf{d}')}+c^{-\Delta_2(D-1,\mathbf{d}')}\leq 1.
\end{align*}
\end{proof}

\paragraph*{Acknowledgements.}
We would like to thank Yixin Cao for valuable discussions and thank organizers of PACE challenge 2016 for motivating us to study \problem{FVS}.

\bibliographystyle{plainurl}
\bibliography{main}

\appendix
\section{Program to check Lemma~\ref{lem:mc}}\label{sec:program}
We attach a source code of a python3 program to evaluate the inequalities (\ref{eq:mc1})--(\ref{eq:mc3}) appeared in the proof of Lemma~\ref{lem:mc}.
The same source code is also available at \url{https://github.com/wata-orz/FVS_analysis}.

\begin{lstlisting}[language=Python,showstringspaces=false,columns=fullflexible]
#!/usr/bin/env python3

a = 0.922863
b = [0.0, 0.0, 0.0, 0.114038, 0.186479, 0.238143, 0.277239, 0.308030, 0.332974,
  0.353536, 0.370540, 0.384771, 0.396884, 0.408715, 0.418855, 0.427643, 0.435333,
  0.442118, 0.448149, 0.453544, 0.458401, 0.462794, 0.466788, 0.470435, 0.473778,
  0.476853, 0.479691, 0.482320, 0.484760, 0.487032, 0.489153, 0.489153]
c = 3.460

eps = 1e-7 # for avoiding rounding errors

# enumerate all d
def all_d(d, i, d_min):
    yield d
    if i < 4:
        for di in range(d_min, 32):
            yield from all_d(d + [di], i + 1, di)

# inequality (1)
for d in all_d([], 0, 3):
    f = len(d)
    d_prime = 4 + sum(d) - 2 * f
    delta1 = 1 - f * a / 4 + sum(map(lambda di: b[di] - b[di-1], d))
    delta2 = a - 2 * a / 4 + sum(map(lambda di: b[di], d)) - b[min(30, d_prime)]
    assert(c**(-delta1) + c**(-delta2) <= 1 - eps)

# inequality (2)
for D in range(5, 32):
    for d_prime in range(D, 32):
        assert(-2 * a / D - b[d_prime] + 2 * a / (D - 1) + b[d_prime - 1] >= eps)

# inequality (3)
for d_f in range(3, 32):
    for d_prime in range(d_f+1, 32):
        assert(b[d_f] - b[d_prime] + b[d_prime - d_f + 2] >= eps)

print("all the inequalities are satisfied")
\end{lstlisting}

\end{document}